\documentclass[prx,aps,10pt,english,superscriptaddress,longbibliography, nobibnotes]{revtex4-2}
\usepackage{graphicx}
\usepackage{amsmath}
\usepackage{amssymb}
\usepackage{amsfonts}
\usepackage{amsthm}
\usepackage{subfigure}
\usepackage{dsfont}
\usepackage{appendix}
\usepackage{tikz}
\usepackage{braket}
\usepackage{listings}
\usepackage{comment}
\usepackage{array}
\usepackage{booktabs}
\AtBeginDocument{\renewcommand{\selectlanguage}[1]{}}
\usepackage[hidelinks,colorlinks=false]{hyperref}

\newtheorem{theorem}{Theorem}

\DeclareMathOperator{\tr}{tr}

\begin{document}
\title{Better Pauli Channel Learning with Maximum Likelihood Estimation}
\def\urbana{
Institute for Condensed Matter Theory and IQUIST and NCSA Center for Artificial Intelligence Innovation and Department of Physics, University of Illinois at Urbana-Champaign, IL 61801, USA
}
\def\ibm{IBM Quantum, IBM T.J. Watson Research Center, Yorktown Heights, NY 10598, USA}

\author{Daniel Belkin}\affiliation{\urbana}
\author{Faisal Alam}\affiliation{\urbana}
\author{Matthew Thibodeau}\affiliation{\urbana}
\author{Alireza Seif}\affiliation{\ibm}
\author{Ewout van den Berg}\affiliation{\ibm}
\author{Bryan K. Clark}\email{bkclark@illinois.edu}\affiliation{\urbana}

\date{\today}
\begin{abstract}
    Error mitigation in a noisy quantum device requires a very good estimate of the noise channel. The accuracy of probabilistic error cancellation is often limited by the high sample complexity of channel tomography. In principle, optimal sample complexity is attained by maximum likelihood estimation (MLE), but MLE is computationally challenging. We show that MLE can be made computationally tractable in certain cases of interest. For the common case of a 1D-local sparse Pauli-Lindblad channel, the likelihood function reduces to an efficiently-evaluable Bayesian network. 
    We show that the resulting computation leads to substantially improved tomography. In addition, we demonstrate by simulation that this can lead to meaningful improvements to the overhead of error mitigation. We also discuss possible extensions of our algorithm to more general settings, such as non-1D circuits and non-Pauli errors.
\end{abstract}

\maketitle

\section{Introduction}
Quantum computers are inherently noisy. Achieving practical quantum advantage requires techniques to robustly correct errors induced by this noise. While quantum error correction is a long-term solution to this problem, current quantum computers are not yet capable of performing fault-tolerant logical computation on many qubits. Instead, error mitigation techniques must be used\cite{filippov_scalable_2023, guo_quantum_2022}. Even near-term applications of quantum error correction (QEC) will be limited in their logical qubit counts and code distances, since physical qubits are limited in number and connectivity\cite{katabarwa_early_2024}. Error-corrected circuits are likely to require additional error mitigation on top of them to fully correct observables, since error mitigation allows one to use classical postprocessing rather than requiring additional physical qubits\cite{wagner_learning_2023, zheng_efficient_2026, zhou_error_2025}.

Probabilistic error cancellation (PEC) stands out as a powerful framework for error mitigation\cite{temme_error_2017, van_den_berg_probabilistic_2023, gupta_probabilistic_2024, chen_disambiguating_2026}.
It offers theoretical guarantees of accuracy and convergence. This comes at the cost of requiring very precise knowledge of the noise model and a large sampling overhead. This stands in contrast to zero-noise extrapolation (ZNE), an alternative approach with low overhead but only heuristic accuracy and convergence\cite{temme_error_2017, li_efficient_2017, kim_evidence_2023, aharonov_reliable_2026}. 

The core idea of PEC is to use the inverse of the noise channel to cancel the noise. This inverse cannot be implemented as a physical quantum channel, but a combination of quantum operations and classical postprocessing can imitate it in expectation. The strategy involves running many variants of the circuit and taking a weighted combination of the results, with weights designed to cancel out the effect of the noise. In order to succeed, however, this interference needs to be very precisely calibrated; PEC is highly sensitive to tomographic precision. 

The difficulty of channel learning depends on the family of noise models considered. Arbitrary channels are intractable to even represent, let alone learn. Realistic assumptions are thus necessary to make any progress possible. One plausible family of noise models is gate-based, where the noise channel depends only on which gate was applied and affects only those qubits on which the gate acts. This sort of model is relatively straightforward to learn\cite{poyatos_complete_1997}, but cross-talk between qubits is usually too important to neglect\cite{sarovar_detecting_2020, van_den_berg_probabilistic_2023}. In the worst case, on the other hand, the structure of the noise could depend on the entire history of the circuit in some complicated way. Such a model would be very difficult to learn. Here we make an intermediate assumption: Noise may be correlated in space but is uncorrelated in time. The noise channel then depends only on each layer of the circuit (i.e., the set of gates applied concurrently)\cite{carignan-dugas_error_2023}. We will also assume that the noise has been Pauli twirled, so that the channel is a Pauli channel, and impose a sparsity constraint to keep the total parameter count polynomial in the number of qubits\cite{geller_efficient_2013, hashim_randomized_2021, wallman_noise_2016, van_den_berg_probabilistic_2023}.

The result of these assumptions is that tomography may be done separately for each layer of the desired circuit. For deep circuits, the sample complexity of tomography adds up to a very significant cost. It follows that even modest improvements to the sample complexity of tomography can result in meaningful reductions in runtime. Beyond error mitigation, channel learning also provides vital diagnostic insights for hardware design, cross-talk characterization, and the optimization of tailored QEC codes\cite{nielsen_gate_2021, nautrup_optimizing_2019, sivak_optimization_2024}. 

The standard approach to learning Pauli channels is based on estimating each Pauli fidelity directly from measurements\cite{flammia_efficient_2020}. We term this the Empirical Pauli Fidelities (EPF) estimator. The EPF estimator has asymptotic convergence guarantees and is relatively natural to construct. It thus may seem that no improvement is possible. In this work, however, we will show that maximum likelihood estimation (MLE) does substantially better. This is because EPF discards important information from the training data. MLE makes use of all the available information, resulting in better sample complexity and significantly more accurate error-mitigated quantum circuits. Figure \ref{fig:overview} gives a schematic depiction of the basic conceptual structure of this work. 

\begin{figure}[h]
    \centering
    \begin{tikzpicture}
        \begin{scope}
            \node[anchor=north west,inner sep=0] (image_a) at (0,0)
        {\includegraphics[width=0.89\columnwidth]{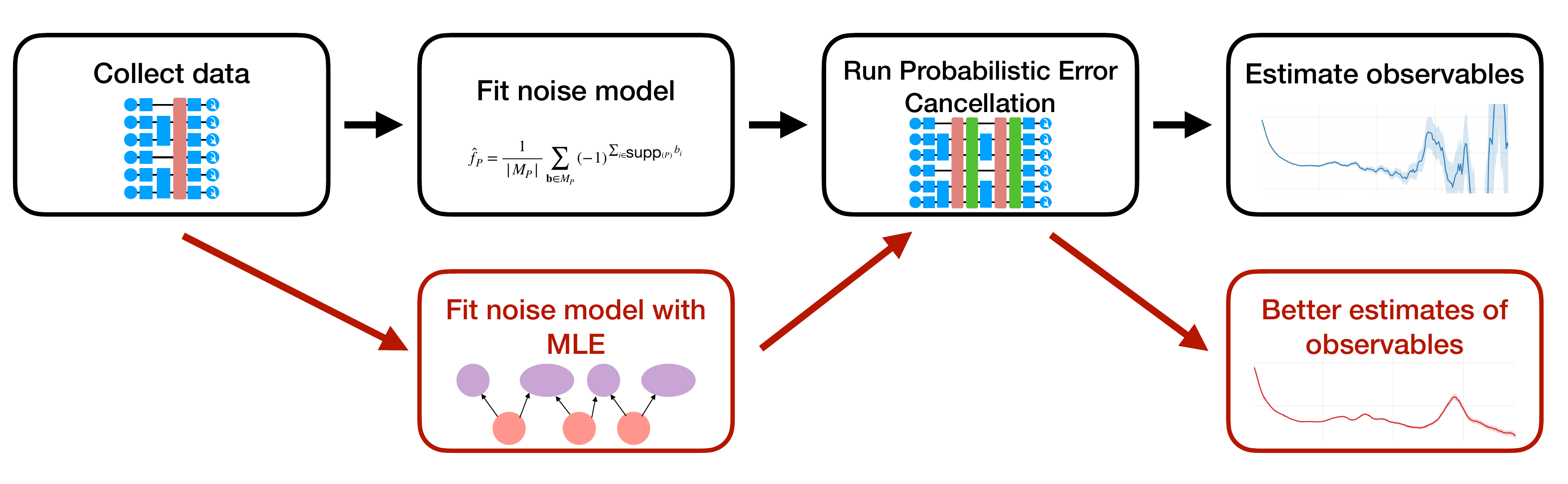}};
        \end{scope}
    \end{tikzpicture}
    \vspace*{-0.4cm}
    \caption{Schematic illustration of our work. We give an algorithm that allows MLE to be used for channel learning. This gives better learning of noise models than previous approaches. When used for PEC on a quantum circuit, the resulting noise model also gives markedly better error-mitigated outputs. 
    }
    \label{fig:overview}
\end{figure}

\subsection{Summary of main results}
PEC requires a precise model of the error channel. In practice, learning the noise model constitutes a large portion of the total overhead associated with PEC. We show that maximum likelihood estimation is meaningfully more sample-efficient than previous approaches to learning Pauli-Lindblad channels. Under realistic conditions, we observe roughly a threefold reduction in sample complexity for learning from the exact same data. Furthermore, as shown in Figure \ref{fig:mle_vs_epf_pec}, this difference translates into large improvements in error-mitigated quantum simulations. 

\begin{figure}[h]
    \begin{tikzpicture}
        \begin{scope}
            \node[anchor=north west,inner sep=0] (image_a) at (0,0)
        {\includegraphics[width=0.8\columnwidth]{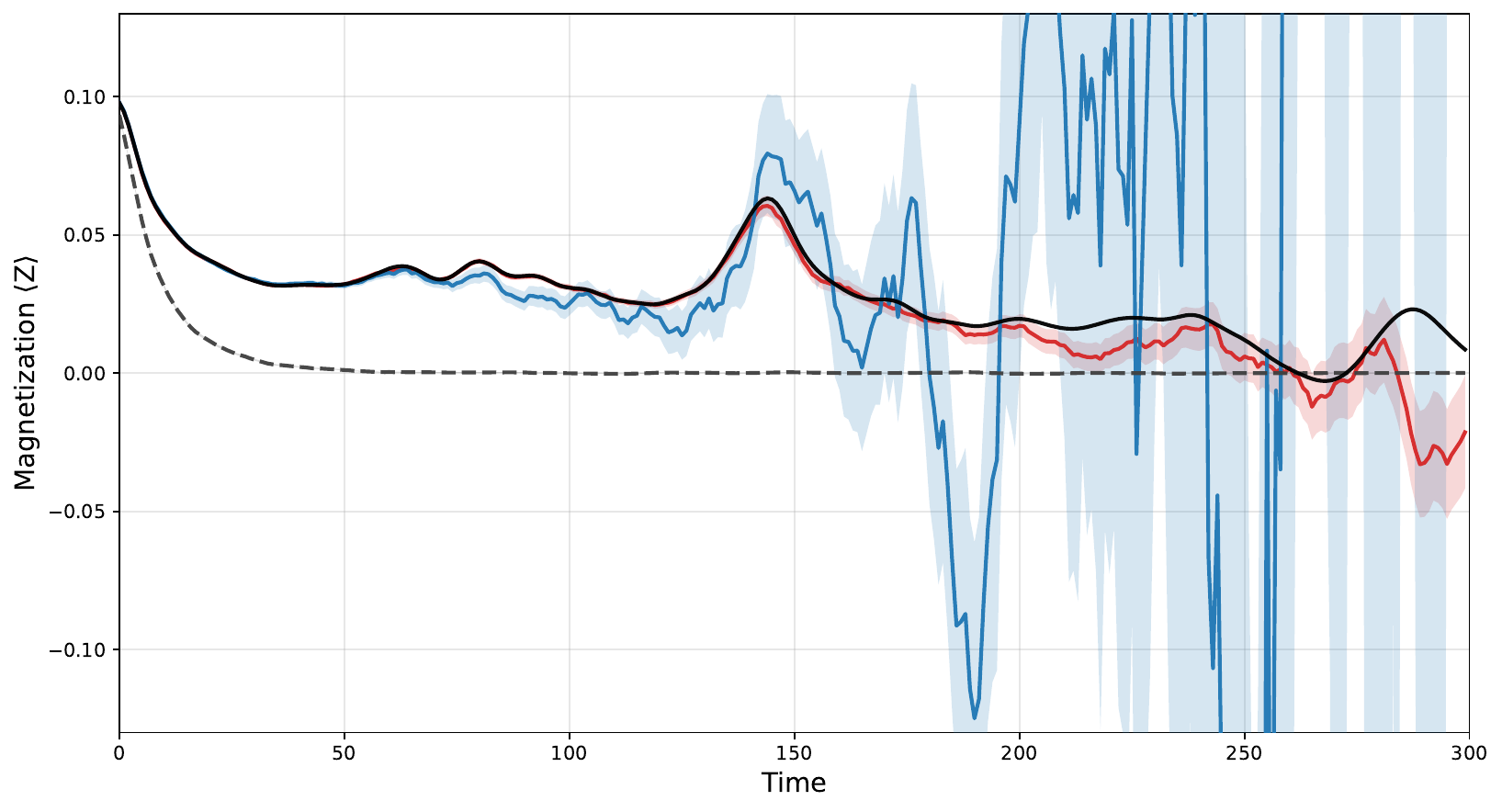}};
        \end{scope}
    \end{tikzpicture}
    \vspace*{-0.4cm}
    \caption{Average magnetization over time in a simulation of a Trotterized transverse-field Ising model circuit on 10 qubits. The noise channel is a 1D 2-local Pauli-Lindblad channel with random coefficients sampled uniformly from $(0,10^{-3})$. PEC is done using channel estimates learned with either MLE or EPF on a training set of 999,999 samples. Shaded regions show standard error. See Section \ref{sec:demo} for further details.
    }
    \label{fig:mle_vs_epf_pec}
\end{figure}

The usual difficulty of MLE is the classical computational cost. A naive algorithm requires exponential runtime. We give an efficient method for evaluating the likelihood function for a 1D-local Pauli-Lindblad channel with data generated by preparing product states and making product measurements. One can obtain a polynomial-time algorithm for this setting using tensor networks, but we find significant practical advantages by instead mapping the quantum circuit to a Bayesian network\cite{dambrosio_inference_1999, xing_probabilistic_nodate}. The likelihood function is then evaluated by belief propagation\cite{pearl_reverend_1982}.

For practical purposes, the main limitation of our approach is the restriction to 1D geometry. We suggest some possible strategies for approximating the likelihood function in higher dimensions, but it is not yet clear how the resulting approximation error will affect convergence. Our algorithm does also work when the channel is attached to a layer of local gates rather than being accessible as a standalone Pauli channel. This setting is especially relevant to PEC.

Because noise rates on quantum hardware are usually quite low, it is useful to amplify the noise by repeating the gates and channel several times before each measurement. So long as the gates are all self-inverse Clifford gates (e.g., CZ gates), our algorithm can be extended to this case. This type of gate is often responsible for a large proportion of the total noise in a quantum computation. In addition, given access to data collected at multiple circuit depths, simultaneous learning of state preparation and measurement (SPAM) error rates and gate-based error rates is possible.

\section{Theoretical background}
Our goal is to learn a quantum channel from experimental data. In the error mitigation context, the most important accuracy measure is the typical bias of whichever error-mitigated observable is of interest. This depends heavily on the details of the problem in question, and in general it is difficult to compute. We will prioritize a proxy measure instead. In this work we generally choose the mean squared difference between learned and true $w_i$ as our measure of success. We expect that our results are not particularly sensitive to the choice of metric. 

A generic quantum channel on $n$ qubits has $O(16^n)$ parameters, so some assumptions about the structure of the channel are necessary. In the error mitigation context, Pauli twirling allows us to reduce arbitrary channels to Pauli channels\cite{hashim_randomized_2021, wallman_noise_2016, geller_efficient_2013}, which can be written in the form
 \begin{gather}
     \Phi(\rho) = \sum_{i} p_i P_i \rho P_i
 \end{gather}
 where the $P_i$ are Pauli strings and the $p_i$ form a probability distribution over errors. This brings the number of parameters down to $O(4^n)$. 

In a physical device, it is natural to assume that the error channel involves a number of statistically independent events, each of which occurs with some characteristic rate. Each event induces a different Pauli error. This assumption corresponds to writing the channel in Pauli-Lindblad form, 
 \begin{gather}
    \label{eq:pl_factorized}
     \mathcal{N}(\rho) = \prod_i \mathcal{N}_{P_i, w_i} 
\end{gather}
with
\begin{gather}
\mathcal{N}_{P_i,w_i}(\rho) = (1 - w_i)\rho + w_i P_i\rho P_i
\end{gather}
$w_i$ is interpreted as the probability of error $P_i$ occurring in one timestep, here meaning one layer of the circuit. The terms in the product in Equation~\eqref{eq:pl_factorized} commute, so their order doesn't matter. We may equivalently expand Equation~\eqref{eq:pl_factorized} as 
\begin{gather}
      \mathcal{N}(\rho) = \sum_{S \subseteq \{1...m\}} \left(\prod_{i \in S} w_i\right) \left(\prod_{i \notin S} (1 - w_i)\right)
    \left(\prod_{i \in S} P_i\right) \rho \left(\prod_{i \in S} P_i\right)^\dagger
 \end{gather}
 where the ordering of the products is again irrelevant, except that the two copies of $\prod_{i \in S} P_i$ must share the same ordering.

 It is further useful to assume that most errors affect only a small number of adjacent qubits. In other words, the support of $P_i$ is small and geometrically local. This is not entirely true, since real devices occasionally experience events affecting many qubits\cite{seif_suppressing_2024, li_cosmic-ray-induced_2025, binney_distinguishing_2026}, but it is an assumption that accounts for most of the noise in practice\cite{acharya_quantum_2025}. With this assumption, the number of parameters which must be learned is now only polynomial in the number of qubits. 

  For simplicity, we will focus on the case in which the noise is 2-local and both the circuit and the noise are geometrically 1D. Parameter count then scales linearly with qubit count. We will furthermore mostly assume that the circuit which produces the noise involves only $CZ$ gates. We will later discuss how these constraints may be relaxed.

\subsection{Data collection}
Data is collected by preparing states, applying the noisy gates, and then measuring in an appropriately selected basis. We focus on the case where the prepared states and measurement basis have a particular simple structure. We prepare on each site an eigenstate of a Pauli operator, and subsequently measure that same Pauli operator on that site. This is illustrated in Figure \ref{fig:schematic_setup}, where $R_i$ represents a rotation to a Pauli eigenbasis on site $i$. 

The remaining question is which product states to prepare and in what basis to measure them. For the context we choose as our main example, 1D 2-local Pauli-Lindblad channels with the idling circuit, Ref~\cite{van_den_berg_probabilistic_2023} shows there exists a set of 9 distinct settings of the $R_i$ which suffice for learnability (Theorem SIV.4). This is the strategy used in the bulk of this paper. We emulate quantum data collection using the Qaravan python package\cite{alam-faisal_alam-faisalqaravan_2026}.

Our algorithm can also be applied to somewhat more general experiments. We assume for now that it is possible to prepare product states and measure in a product basis without errors. Section \ref{sec:measurement_error} shows how this assumption can be removed. Section \ref{sec:cycling} shows how to generalize to the case in which the gates and associated noise channel are applied many times rather than once. One can also generalize to other state preparations, e.g., non-Pauli rotations. In principle one may be able to improve sample complexity by preparing states and making measurements which are entangled across multiple runs of the circuit\cite{chen_optimal_2024, chen_instance-optimal_2026}. This is challenging on current hardware\cite{chen_when_2023}, however, and is not considered here.

\begin{figure}[h]
    \centering
    \begin{tikzpicture}
        \begin{scope}
            \node[anchor=north west,inner sep=0] (image_a) at (0,0)
            {\includegraphics[width=0.4\columnwidth]{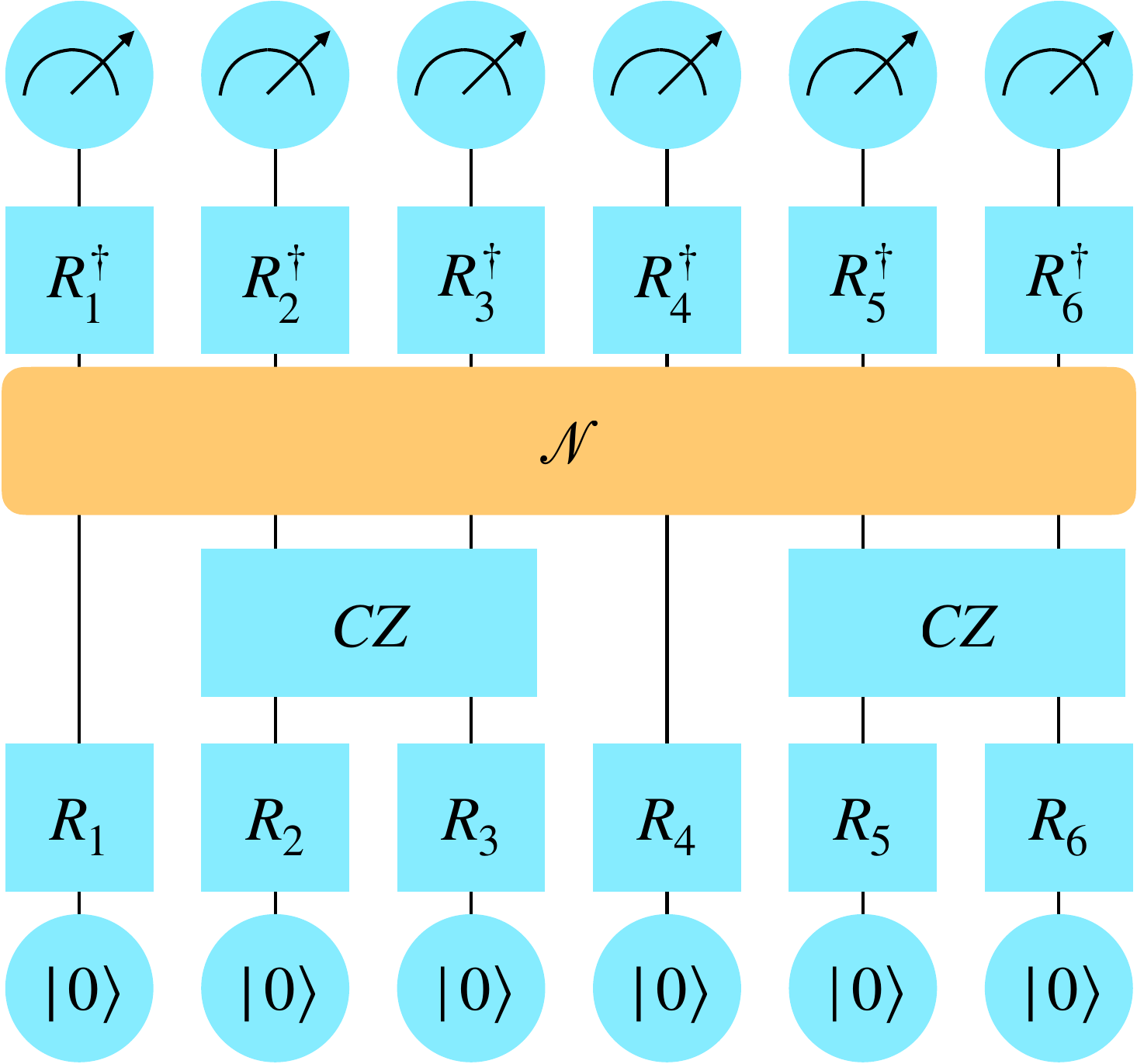}};
        \end{scope}
    \end{tikzpicture}
    \vspace*{-0.4cm}
    \caption{Schematic illustrating a typical data collection setup for learning gate-based Pauli errors. A layer of $CZ$ gates induces noise $\mathcal{N}$. We assume noiseless single-qubit rotations $R_i$ can be used to prepare states and rotate to a desired measurement basis. 
    }
    \label{fig:schematic_setup}
\end{figure}

\subsection{Empirical Pauli Fidelities}
Ref.~\onlinecite{erhard_characterizing_2019} introduces an approach to tomography which we term Empirical Pauli Fidelities (EPF). The strategy is to choose a particular parameterization of the space of channels such that there exists an unbiased estimator for the parameters. The central limit theorem then guarantees a consistent estimator with typical error $\sim N_\text{samples}^{-1/2}$\cite{flammia_efficient_2020}. It turns out that the Pauli fidelities \(f_{P} = \tr\left( P \mathcal{N}\left[P\right]\right)\) are such a parameterization. Here $\mathcal{N}$ is the channel and $P$ is any Pauli string. The unbiased estimator is given by 
\begin{gather}
    \hat{f}_{P} = \frac{1}{|M_{P}|}\sum_{\mathbf b \in M_{P}} (-1)^{\sum_{i \in \text{supp}(P)} b_i}
\end{gather}
where  $M_P$ is the set of measurements in bases containing $P$ as a substring and $\text{supp}(P)$ is the support of $P$. See Ref~\cite{flammia_efficient_2020} for a detailed discussion of this algorithm.

These algorithms are very computationally efficient. Furthermore, the resulting estimate is unbiased with respect to some parameterization. However, we will show that their sample complexity is not optimal.

It is clear that these algorithms proceed by taking the Empirical Pauli Fidelities as a summary of the data, then trying to fit the best channel given that summary statistic. An immediate question raised is whether the EPFs are a ``good'' summary, i.e., whether they capture the important information in the data. In Appendix \ref{app:insufficient}, we show that they do not. In particular, we construct a simple example in which the EPF summary statistic can be seen to discard important aspects of the raw data. It follows that the EPFs are not a sufficient statistic; they do not capture all information about the likelihood function. One thus expects to obtain better estimates of the true channel parameters using estimators that instead exploit \textit{all} of the relevant information in the data.

\subsection{Maximum Likelihood Estimation}
The standard approach to minimizing the sample complexity of tomography is known as Maximum Likelihood Estimation (MLE). MLE is \textit{asymptotically efficient}, which is to say the variance of the estimate is the best possible in the limit of large sample size. This guarantees that MLE will give more accurate estimates than EPF when enough data is available. However, computational evaluation of the MLE estimator is not always straightforward. 

Suppose we prepare a quantum state $\rho_i$ and pass it through a quantum circuit $C$ which induces an unknown noise channel $\Phi$. We then make a projective measurement and observe outcome $O_f$. The corresponding likelihood function is
\begin{gather}
    \mathcal{L}( \Phi | O_f) = \text{tr}\left(O_f \Phi\left( C \rho_iC^\dagger\right)\right)
\end{gather}
Likelihood for a larger dataset of several observations is the product of the single-sample likelihoods. For a generic channel,  $\Phi(U \rho_iU^\dagger)$ requires exponential space to represent, so this formula is not practical to evaluate. And even if we could evaluate the likelihood function, the space of possible channels is still too large to have any hope of successfully optimizing. 

For Pauli channels, however, things are much easier. In our case $\rho_i = R \ket{\mathbf{0}}\bra{\mathbf{0}}R^\dagger$ and $O_f = R \ket{\mathbf{b}}\bra{\mathbf{b}}R^\dagger$, where $R$ is some Pauli basis rotation and $\mathbf{b}$ is the bitstring obtained after measurement. The likelihood function then simplifies to 
\begin{gather}
    \mathcal{L}( p_1...p_m | \mathbf{b}) = \sum_{i \in \text{errors}} p_i \left|\bra{\mathbf{b}} R^\dagger P_i C R \ket{\mathbf{0}}\right|^2
\end{gather}
This sum is computationally efficient to evaluate so long as the number of terms in the channel is not too large, i.e., when the Pauli channel is sparse. 

We may also consider a sparse Pauli-Lindblad channel. Although a Pauli-Lindblad channel is a special kind of Pauli channel, a \textit{sparse} Pauli-Lindblad channel is not a \textit{sparse} Pauli channel. This is because the corresponding Pauli channel has one term for each possible combination of generators, so that a small number of Lindblad generators produces a large number of possible Pauli errors.

For the Pauli-Lindblad channel, we may write the likelihood function as 
\begin{gather}
    \label{eq:spl_likfun}
    \mathcal{L}(w_1...w_m | \mathbf{b}) = \sum_{S \subseteq \{1...m\}} \left(\prod_{i \in S} w_i\right) \left(\prod_{i \notin S} (1 - w_i)\right)
    \times
    \left|\bra{\mathbf{b}} R^\dagger \left(\prod_{i \in S} P_i\right) C R \ket{\mathbf{0}}\right|^2
\end{gather}
This is a sum over all possible \textit{sets} of basis errors which may co-occur. Even when the number of Lindblad generators is polynomial in $n$, this sum still has an exponential number of terms, so it is difficult to evaluate directly. However, the dimension of the space of channels in this case is small, so the optimization itself is not hopeless. We will later present some tools which make it tractable to evaluate the sum.

\section{Algorithms}

\subsection{Sparse Pauli channel}
This case is quite straightforward. Suppose our possible measurement settings are labeled by basis rotations $R_j$. The probability of error string $P_i$ occurring is $w_i$, and the likelihood of measurement outcome $\mathbf{b}_j$ conditional on error $P_i$ is $\left|\bra{\mathbf{b}_j} R_j^\dagger P_i C R_j \ket{\mathbf{0}}\right|^2$, so the likelihood function is 
\begin{gather}
    \mathcal{L}( w_1...w_m) = \prod_{j \in \text{data}} \sum_{i \in \text{errors}} w_i \left|\bra{\mathbf{b}_j} R_j^\dagger P_i C R_j \ket{\mathbf{0}}\right|^2
\end{gather}
Since the circuit $C$ is a tensor product of independent gates, each term in this sum can be evaluated easily, and the total number of terms is small. One can then use any optimization strategy to look for the maximum.

\subsection{Tensor network for Pauli-Lindblad channel}
 Because the errors $P_i$ are local, each $\mathcal{N}_{P_i, w_i}$ can be represented by a small tensor. A sparse Pauli-Lindblad channel is thus a product of small tensors. The Pauli basis rotations $R$ are single-site operators, the initial state is a product state, and the final measurement is in a product basis, and so the whole likelihood of each measurement can be represented efficiently by a single tensor network (Figure \ref{fig:network_reductions}a). When the geometry is 1D, this network can be contracted efficiently and the likelihood function can be evaluated in polynomial time. 

Ref.~\onlinecite{cao_differentiable_2026} uses tensor networks in a QEC setting to learn a noise model from syndromes. In that case the geometry is not 1D, which limits scalability. Even in the 1D case, we found the constant-factor overhead of tensor network MLE is problematic. We will now discuss an alternative strategy with lower overhead and potentially better generalization beyond 1D.

\subsection{Bayesian network for Pauli-Lindblad channel}
We will now give a reduction from the more general tensor network representation of the likelihood function to a Bayesian network. In other words, we reinterpret this quantum problem as describing in effect the evolution of a classical probability distribution over computational basis states. In the 1D case, the asymptotic computational complexity of evaluation of the Bayesian network is the same as that of the tensor network, but the constants seem to be much better. In addition, the Bayesian network provides a conceptual picture which is useful for generalizing to non-1D systems.

The first observation is that we may rewrite 
\begin{align}
    \bra{\mathbf{b}} R^\dagger \left(\prod_{i} P_i\right) C R
    \ket{\mathbf{0}} 
    &= \bra{\mathbf{b}} \left(\prod_{i} R^\dagger P_i R \right) R^\dagger C R
    \ket{\mathbf{0}} \\
    &= \bra{\mathbf{b}} \left(\prod_{i} P_i'\right) C'
    \ket{\mathbf{0}}
\end{align}
Here $R$ is a Clifford gate, so $P_i' \equiv R^\dagger P_i R$ remains some Pauli string of small support. $C' \equiv R^\dagger C R$ remains a single layer of two-site gates, but the gates may no longer be CZs. This argument carries through to the whole error channel, as illustrated in Figure \ref{fig:network_reductions}b.

\begin{figure*}[h]
    \centering
    \begin{tikzpicture}
        \begin{scope}
            \node[anchor=north west,inner sep=0] (image_a) at (0,0)
            {\includegraphics[width=0.95\columnwidth]{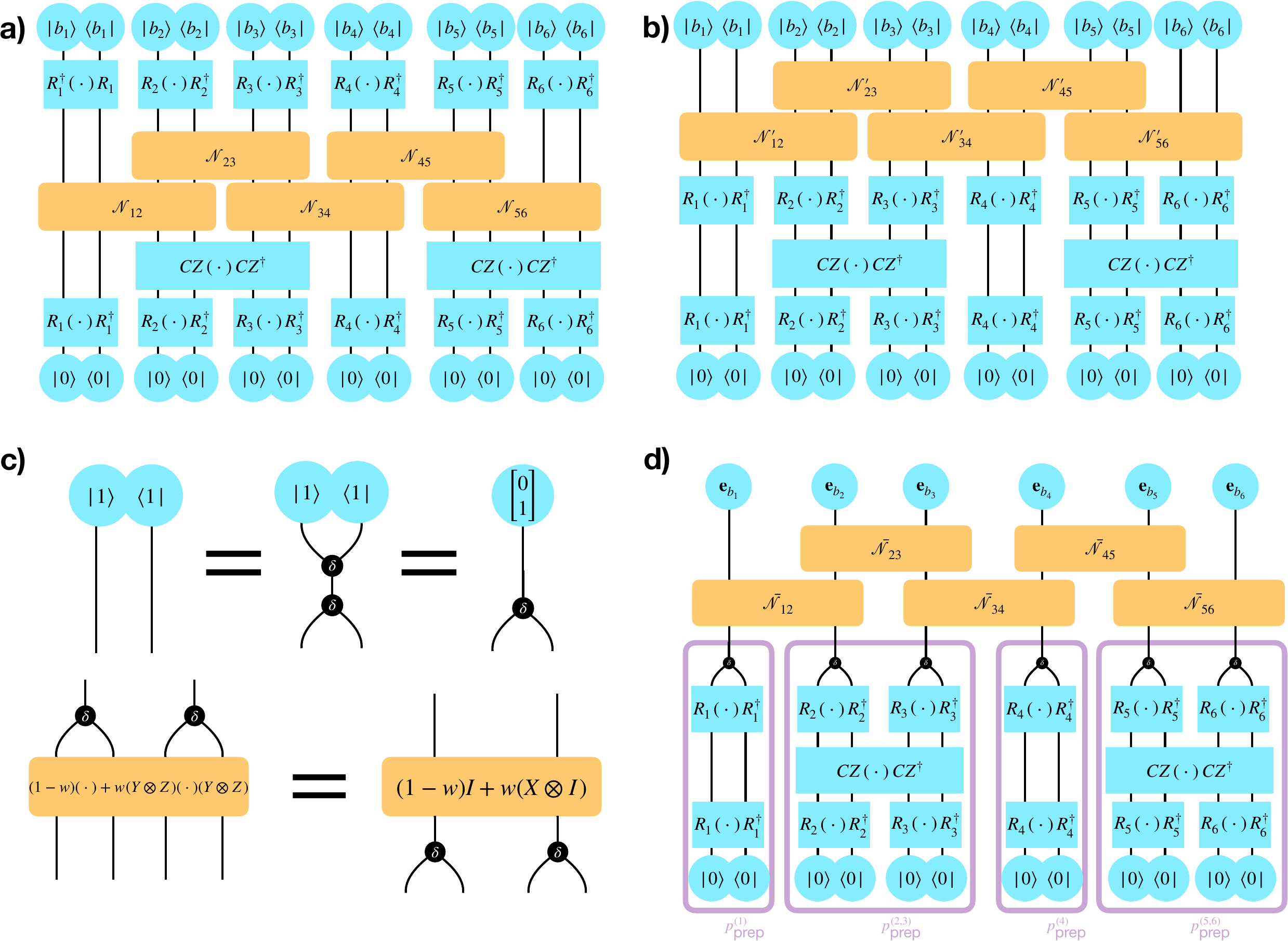}};
        \end{scope}
    \end{tikzpicture}
    \vspace*{-0.4cm}
    \caption{(a) Original tensor network. (b) After commuting the basis rotations through the errors. (c) Reduction rules map quantum states to classical probability distributions. (d) Error channel as evolution of a classical distribution.}
    \label{fig:network_reductions}
\end{figure*}

The second observation is that in the computational basis, 
\begin{gather}
    \left(\ket{0}\bra{0}\right)_{ij} = \sum_{k} \delta_{ijk} \left(
    \begin{bmatrix}
        1 \\
        0
    \end{bmatrix}
    \right)_k
\end{gather}
where 
\begin{gather*}
    \delta_{ijk} = \begin{cases}
        1 & i = j = k \\
        0 & \text{otherwise}
    \end{cases}
\end{gather*}
More generally
\begin{gather}
    \left(\ket{b}\bra{b}\right)_{ij} = \sum_{k} \left(
    \mathbf{e}_{b}
    \right)_k \delta_{ijk}
\end{gather}
where $(\mathbf{e}_{b})_k = \delta_{kb}$ represents a standard basis vector. 
This rule is illustrated in Figure \ref{fig:network_reductions}c. The three-point $\delta_{ijk}$ can be interpreted as the decoherence operation taking a quantum density operator to a classical probability distribution over bit strings. For the action of Pauli errors, similarly, we have
\begin{align}
        \sum_{ij}  \delta_{ijk} (Z \rho Z^\dagger)_{ij} &= \sum_{ij}  \delta_{kij} (I \rho I^\dagger)_{ij} = \sum_{ij\ell}  I_{k\ell} \delta_{\ell ij} \rho_{ij}
        \\
        \sum_{ij}  \delta_{ijk} (X \rho X^\dagger)_{ij} &= \sum_{ij}  \delta_{ijk} (Y \rho Y^\dagger)_{ij} = \sum_{ij\ell} X_{k \ell} \delta_{\ell ij} \rho_{ij}
\end{align}
In other words, quantum $X$ and $Y$ operators correspond to classical bit flip ($X$) operations, while quantum $I$ and $Z$ have no effect on the classical distribution. We can extend this argument to Pauli channels as follows: Given a Pauli string P, define $\bar{P}$ to be a Pauli string with $I$ where $P$ has $I$ or $Z$ and with $X$ where $P$ has $X$ or $Y$. Then we may write
\begin{gather}
    \sum_{\vec{i},\vec{j}} \delta_{k_1 i_1 j_1}\delta_{k_2 i_2 j_2} ... \delta_{k_n i_n j_n} \left[(1 - w)\rho + w P \rho P^\dagger\right]_{\vec{i},\vec{j}} 
    =
    \sum_{\vec{\ell}, \vec{i}, \vec{j}} \left[(1 - w)I^{\otimes n} + w\bar{P}\right]_{\vec{k}, \vec{\ell}} \delta_{\ell_1 i_1 j_1} \delta_{\ell_2 i_2 j_2} ... \delta_{\ell_n i_n j_n}
    \rho_{\vec{i},\vec{j}}
\end{gather}
or more compactly
\begin{gather}
   \delta^{\otimes n} \left((1 - w)\rho + w P \rho P^\dagger\right)
    =
    \left((1 - w)I^{\otimes n} + w\bar{P}\right) \delta^{\otimes n} \rho
\end{gather}
This is illustrated in Figure \ref{fig:network_reductions}c. More generally, given a Pauli channel 
\begin{gather}
    \mathcal{N} = \prod_i \left((1 - w_i)(\cdot) + w_iP_i (\cdot) P_i\right)
\end{gather}
we define the classical bitflip channel
\begin{gather}
    \label{eq:bitflip_channel}
    \bar{\mathcal{N}} = \prod_i \left((1 - w_i)I + w_i\bar{P}_i\right)
\end{gather}
where this latter product is matrix multiplication rather than channel composition. 

We may apply these rules to reduce the quantum tensor network to a classical tensor network describing the preparation and evolution of a probability distribution, as illustrated in Figure \ref{fig:network_reductions}d. The circuit $C'$ prepares some probability distribution over bitstrings $p_\text{prep}(\mathbf{b})$. The errors then evolve $p_\text{prep}$ according to $\bar{\mathcal{N}}'$. In order to obtain the likelihood of a particular measurement outcome, we compute the probability of the observed bitstring under this distribution. 

\subsubsection{\texorpdfstring{Structure of $p_\text{prep}$}{Structure of pprep}}
\label{sec:pprep}
This probability distribution still requires $2^n$ real numbers to specify, so we haven't improved computational tractability very much yet. We now study the structure of $p_\text{prep}$. Recall that $C$ is a single layer of CZ gates. Consider first a qubit that is not acted on by any of the gates. This qubit is prepared in a state $R^\dagger R \ket{0} = \ket{0}$, so the classical measurement outcome is always $0$. 

Consider now a pair of qubits acted on by a CZ, with rotations $R_1 \otimes R_2$. There are nine possible choices of two-qubit Pauli basis rotation. It turns out when either of the basis rotations is $I$, the prepared state is simply $\ket{00}$, and so the classical measurement outcomes are both always $0$. When both basis rotations are nontrivial, on the other hand, the classical measurement outcomes are two independent $\text{Bernoulli}\left(\frac{1}{2}\right)$ random variables. In the latter case nothing can be learned from measurements of those two bits; the outcome distribution will not depend on the parameters of the noise model at all. Such a basis should thus be avoided when taking data. 

We see that, no matter which Pauli basis rotations we choose, $p_\text{prep}$ is a product distribution
\begin{gather}
    p_\text{prep}(\mathbf{b}) = \prod_{i \in \text{qubits}} p_\text{prep}^{(i)}(b_i)
\end{gather}
Correlations between bits are introduced only by the errors. In order to understand the structure of these correlations, consider a set of Bernoulli random variables 
\begin{gather}
    \ell_{i} = \begin{cases}
        0 & \text{with probability } 1 - w_i \\
        1 & \text{with probability } w_i
    \end{cases}
\end{gather}
Then we can write $(1 - w_i)I^{\otimes n} + w_i\bar{P}_i = \mathbb{E} \left[\bar{P}_i^{\ell_i}\right]$
and the channel of equation~\eqref{eq:bitflip_channel} becomes
\begin{gather}
    \bar{\mathcal{N}} = \mathbb{E}_{\vec{\ell}}\left[\prod_i \bar{P}_i^{\ell_i}\right]
\end{gather}
We can interpret these latent variables $\ell_i$ as indicating which errors actually occur in any given circuit run.

Before any errors have occurred, the prepared distribution $p_{\text{prep}}$ is independent over the bits. After the errors, the final distribution is \textit{conditionally independent}, conditioning on which errors occurred, i.e.,
\begin{gather}
\label{eq:full_bayesian_network}
    p_{\text{out}}(\mathbf{b}) = \mathbb{E}\left[\prod_{i \in \text{qubits}} p_{\text{out}}^{(i)}\left(b_i | \ell_1 ... \ell_m\right)\right]
\end{gather}
In fact, the distribution of any given bit does not depend on most of the latent variables. $\ell_i$ affects bit $j$ only if $\bar{P}_i^{(j)} = X$.
\begin{gather}
    p_{\text{out}}^{(i)}\left(b_i | \ell_1 ... \ell_m\right) = p_{\text{out}}^{(i)}\left(b_i | \left\{\ell_j : i\ \in \text{supp}(\bar{P}_j)\right\}\right)
\end{gather}
This dependence structure can be visualized as a Bayesian network, as illustrated in Figure \ref{fig:bayesian_networks}b. We will give an explicit form for the conditional distributions in Equations~\eqref{eq:conditional_distribution_factors} and~\eqref{eq:full_conditional_distribution}. 

\begin{figure*}[h]
    \centering
    \begin{tikzpicture}
        \begin{scope}
            \node[anchor=north west,inner sep=0] (image_a) at (0,0)
            {\includegraphics[width=\columnwidth]{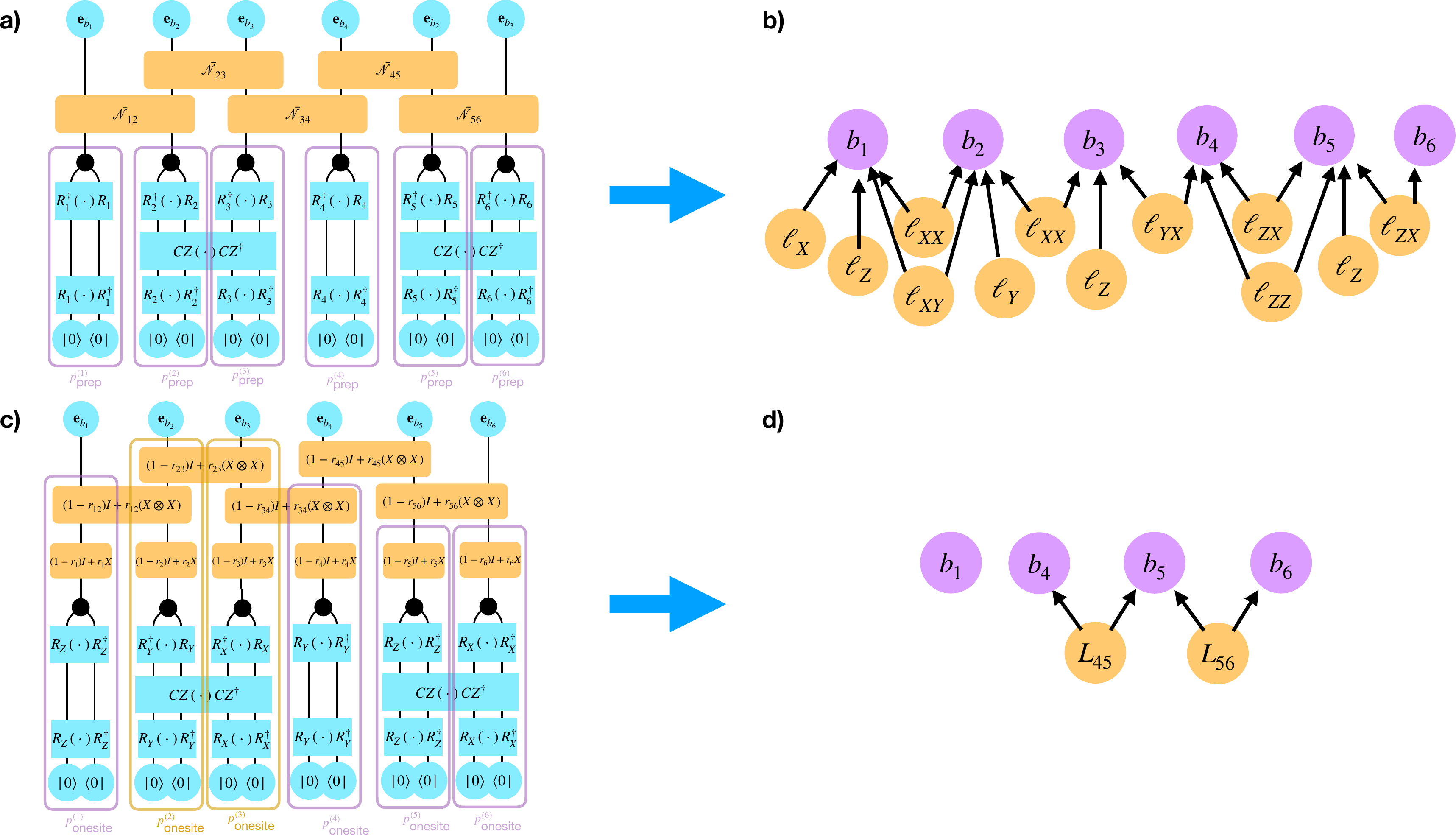}};
        \end{scope}
    \end{tikzpicture}
    \vspace*{-0.4cm}
    \caption{(a) Tensor network describing classical probability distribution preparation and measurement. (b) Bayesian network representation of the same distribution. (c) Simplified tensor network with single-site errors separated out and uninformative bits outlined in yellow. (d) Corresponding reduced Bayesian network.
    }
    \label{fig:bayesian_networks}
\end{figure*}

\subsubsection{Simplifying the network}
Equation~\eqref{eq:full_bayesian_network} describes a Bayesian network. However, in general such a network can't be efficiently evaluated, since $\mathbb{E}_{\vec{\ell}}$ involves a sum over all $2^m$ possible configurations of the latent variables. We will now show that this description of the network has a great deal of redundancy. 

\paragraph{Uninformative bits} Suppose a bit is prepared in an equal mixture of $0$ and $1$, i.e., $p_{\text{prep}}^{(i)}(b) = \frac{1}{2}$. Suppose a $CZ$ gate acts on two sites for which $R_1, R_2$ are both nontrivial. Then $p_{\text{prep}}(b) = \frac{1}{2}$ for both bits, i.e., they are each in an equal mixture of $0$ and $1$. It follows that these bits will be independent of the others at all later times. No matter what latent variables act, they will end up in a distribution $p_{\text{out}}(b) = \frac{1}{2}$. These bits are termed ``uninformative'', since measurement outcomes on them carry no information at all about the structure of the channel. We may drop them from our network entirely. Latent variables that affect only uninformative bits are likewise dropped. The surviving bits all satisfy $p_\text{prep}^{(i)}(b) = \delta_{b0}$. 

\paragraph{Consolidation of error strings} 
Equation~\eqref{eq:full_bayesian_network} introduces one latent variable $\ell_i$ for each Lindblad generator $P_i$. However, the mapping from $P_i$ to $\bar{P}_i$ is not injective, so sometimes distinct Lindblad generators produce the same classical bit flips. Let us instead label classical errors by the subset of bits which are flipped, so e.g 
$\bar{P}_{\{2,4\}} = IXIXII$. For any subset $s$, the aggregate error probability is then given by
\begin{gather}
    r_s \equiv \frac{1}{2}\left(1 - \prod_{\{i:\bar{P}_i = \bar{P}_s\}} (1 - 2 w_i)\right)
\end{gather}
since each $\bar{P}$ is self-inverse. Call the corresponding latent variables $L_s$. This typically reduces the number of latent variables needed dramatically.

\paragraph{Single-qubit errors} It is useful to absorb the single-qubit errors into the state preparation, since this reduces the total number of nodes in the network. More precisely, define
\begin{gather}
    p_{\text{onesite}}^{(i)}(b_i) = \mathbb{E}_{L_{\{i\}}} \left[p_{\text{prep}}^{(i)}\left(L_{\{i\}} \oplus b_i\right)\right]
\end{gather}
or equivalently, since we've already dropped uninformative bits, 
\begin{gather}
    p_{\text{onesite}}^{(i)}(b_i) = \begin{cases} 
    1 - r_{\{i\}} & b_i = 0 \\
    r_{\{i\}} & b_i = 1
    \end{cases}
\end{gather}
Figure \ref{fig:bayesian_networks}d illustrates the effect of these simplifications.

\subsubsection{Explicit form of conditional distribution}
We know that the conditional distribution decomposes over qubits. We may write explicitly 
\begin{gather}
    \label{eq:conditional_distribution_factors}
    p_{\text{out}}^{(i)}\left(b_i | \mathbf{L}^{(i)}\right) = p_\text{prep}^{(i)}\left[\left(\bigoplus_{s : i \in s} \mathbf{L}_s \right) \oplus b_i\right]
\end{gather}
where $\oplus$ is the elementwise XOR. In other words, the conditional distribution of $b_i$ depends only on whether the number of active latent variables touching site $i$ is odd or even. 

If we instead absorb the single-qubit errors, we obtain
\begin{gather}
    \label{eq:full_conditional_distribution}
    p_{\text{out}}(\mathbf{b}) = \mathbb{E}_\mathbf{L}\left[\prod_{i \in \text{qubits}} \begin{cases} 
    1 - r_{\{i\}} & \left(\bigoplus_{s:i \in s \land |s| > 1} L_s\right) \oplus b_i = 0 \\
    r_{\{i\}} & \left(\bigoplus_{s:i \in s \land |s| > 1} L_s\right) \oplus b_i= 1
    \end{cases}\right]
\end{gather}
These conditional distributions are what is needed to construct the Bayesian network. 

\subsubsection{Optimization and convergence}
We can use belief propagation to evaluate the likelihood function at a single point\cite{yedidia_understanding_2003}. Optimization is made more efficient by access to gradients of the likelihood function. We use the jax python library to obtain gradients and the scipy implementation of L-BFGS-B to do our optimization\cite{virtanen_scipy_2020, liu_limited_1989}. Although the likelihood function is not convex, empirically the optimization does not seem to be difficult. There is a vanishing gradient problem near $w_i = 0.5$, but this is easily avoided by initializing all parameters to small values. Furthermore, we find that the estimator approaches the Cram\'{e}r-Rao bound for sufficiently large sample size. This suggests that the optimization successfully locates the global minimum. 

\section{Results}
\subsection{MLE improves sample complexity}
\begin{figure}[h]
    \centering
    \begin{tikzpicture}
        \begin{scope}
            \node[anchor=north west,inner sep=0] (image_a) at (0,0)
            {\includegraphics[width=0.4\columnwidth]{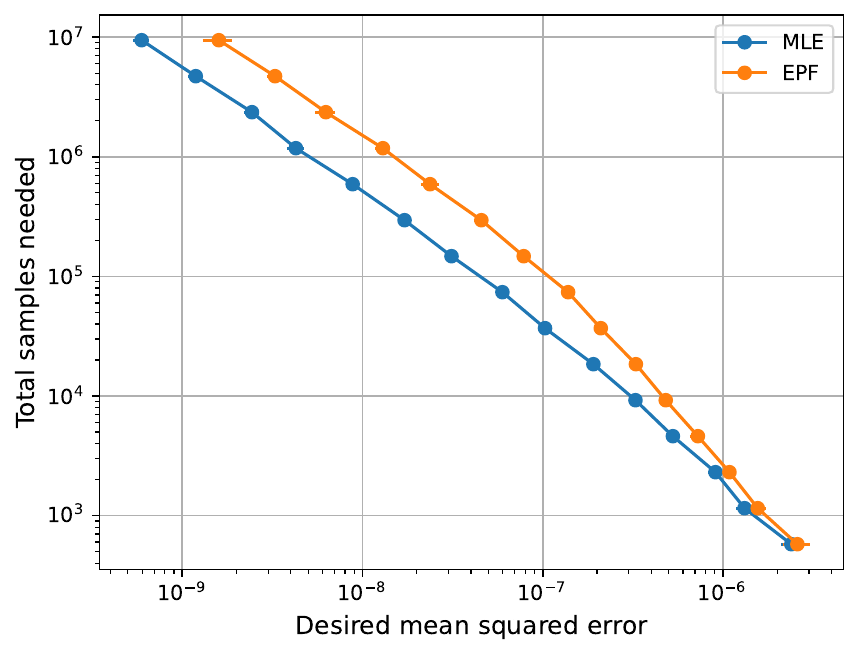}};
            \node [anchor=north west] (note) at (-0.5,0) {\small{\textbf{a)}}};
        \end{scope}
    \end{tikzpicture}
    \begin{tikzpicture}
        \begin{scope}
            \node[anchor=north west,inner sep=0] (image_b) at (0,0)
            {\includegraphics[width=0.43\columnwidth]{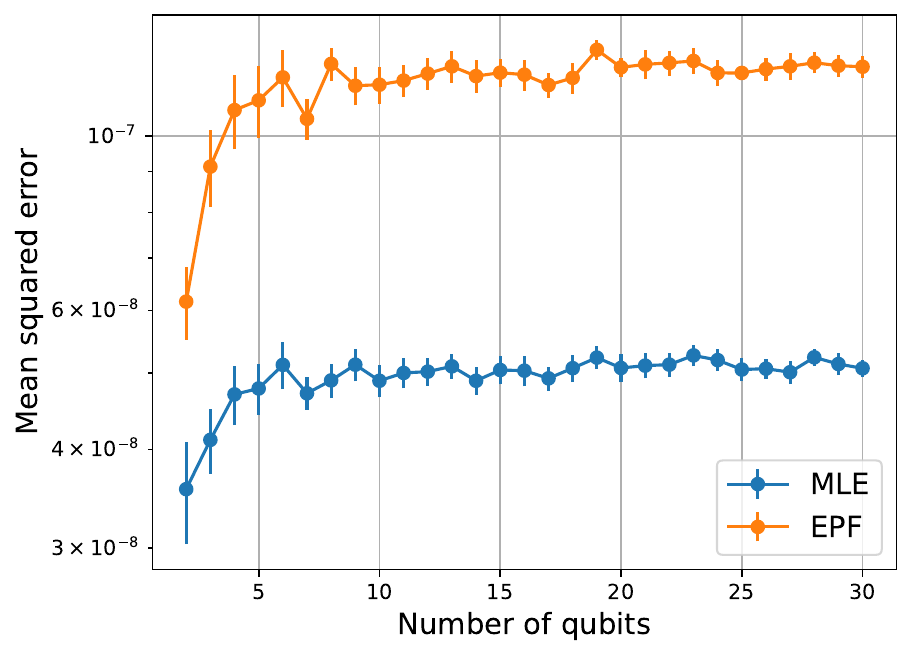}};
            \node [anchor=north west] (note) at (-0.5,0) {\small{\textbf{b)}}};
        \end{scope}
    \end{tikzpicture}
    \vspace*{-0.4cm}
    \caption{
    (a) MLE reaches the same accuracy as EPF with roughly one third the sample size. Here the channel is a 10-qubit 1D-local sparse Pauli-Lindblad channel. $w_i$ are sampled randomly from the uniform distribution on $(0,10^{-3})$. The vertical axis shows mean squared error per parameter, averaged both over the 111 parameters and over many rounds of simulated learning. Error bars indicate standard errors. (b) The sample complexity appears to be nearly constant as system size increases. Here we use 99,999 samples (11,111 for each of the 9 measurement settings). We observe nearly constant per-parameter accuracy. Each point is an average over many channels, each sampled from the same distribution as (a).
    }
    \label{fig:mle_vs_epf_pl}
\end{figure}
Figure \ref{fig:mle_vs_epf_pl} shows comparisons between channel learning with MLE and EPF. Simulated data is gathered according to the model in Figure \ref{fig:schematic_setup}. We measure accuracy by the mean squared difference in our learned estimate of the probability of each generator and the generator's true probability, averaged both over generators and over many trials:
\begin{gather}
    MSE =\left \langle\frac{\sum_i \left(w_i^{\text{estimate}} - w_i^{\text{true}}\right)^2}{\text{\# generators}}\right \rangle_{\text{many trials}}
\end{gather}
Panel (a) shows the number of samples needed to reach a desired accuracy for a particular 10-qubit 1D-local sparse Pauli-Lindblad channel. The channel was generated randomly, with each $w_i$ chosen from the uniform distribution between $0$ and $0.001$. We see that EPF requires approximately three times the sample count in the high-precision regime, with larger advantages in the low-precision (small-sample-size) regime. Both estimators are asymptotically normal, with mean squared error $\sim \frac{1}{\text{sample size}}$. 

Panel (b) shows how accuracy depends on system size. In this panel, errors are averaged over many random channels generated by the same procedure described above. We see that the accuracy is nearly independent of the system size at fixed sample counts (99,999 samples in this case). The advantage of MLE over EPF seems to be independent of the circuit width. This is as expected, since local errors are learned mostly from local data. The problem of learning a channel on a very large system can thus be decomposed into many nearly-independent local subproblems.

\subsection{Improved sample complexity results in meaningfully improved PEC}
\label{sec:demo}
We have seen that MLE gives better sample complexity than previous approaches. We now see that this corresponds to better convergence in a realistic error-mitigated circuit. Figure \ref{fig:mle_vs_epf_pec} showed the simulation of a noisy quantum circuit with PEC using a learned channel. The circuit implements a Trotter simulation of the transverse-field Ising model on 10 qubits. We generate a 1D 2-local Pauli-Lindblad channel with random coefficients distributed uniformly between $0$ and $10^{-3}$. We then draw a sample of $111,111$ outcomes for each of the $9$ measurement settings. Both EPF and MLE are used to fit channels to this data, obtaining mean squared errors per parameter of $2.2 \times 10^{-8}$ and $6.0 \times 10^{-9}$, respectively. For this example we take the noise to be the same for both Trotter layers. Learning is done without any CZ gates, although as discussed in Section \ref{sec:pprep} the difference is slight. 

We plot the evolution of the expected magnetization in four different scenarios. The magnetization shown is both an average over the 10 qubits and an expected value over samples from the error-mitigated circuit. In the ideal, noiseless case, the magnetization oscillates. Unmitigated noise damps the magnetization towards zero over time. For sufficiently short times, both EPF and MLE track the physics reasonably well. At sufficiently long times, error in the learned channel overcorrects for the noise, creating growing modes which diverge substantially from the true dynamics. 

An important ingredient in this behavior is the difference between underestimating and overestimating noise rates. Underestimated parameters leave some residual unmitigated error, which tends to drive the state towards maximally mixed. Overestimated parameters, however, produce an unphysical state, which can result in magnetizations far exceeding one. Indeed, the systematic error introduced by overestimated parameters seems to grow exponentially. It is thus useful to look at the largest single overestimate. In this case it is $1.9 \times 10^{-4}$ with MLE and $3.7 \times 10^{-4}$ with EPF, so we might expect the EPF-mitigated curve to diverge from the ground truth about twice as quickly as the MLE-mitigated curve. Indeed, we see that the EPF curve differs visibly around $t = 70$, while the MLE curve remains accurate until around $t = 180$. 

\section{Extensions}
\subsection{Cycling}
\label{sec:cycling}
Single-gate error rates are generally quite low. If we run just one layer of gates before measuring, we will usually detect no errors at all. A typical data point obtained from the setup in Figure \ref{fig:schematic_setup} thus contains very little information about the error probabilities. 

This problem can be solved by applying a deep circuit before measuring. A deep circuit contains many gates, and so we expect to see many errors. A deep circuit of course takes more time to run. However, in a typical experiment, state preparation and measurement overhead significantly exceeds any one layer's contribution to runtime\cite{noauthor_workload_nodate}. Repeating our layer of gates several times before measuring can thus give large improvements in sample complexity at the cost of only a little additional time per sample\cite{flammia_efficient_2020, harper_statistical_2019}. Appendix \ref{app:fisher} shows that the optimal depth at which to measure seems to be proportional to $\frac{1}{w_i}$. 

We now show that our reduction to a Bayesian network is still useful when a layer of $CZ$ gates is repeated several times. Let $L$ be the set of pairs of sites acted on by our gates, so that the operator for the layer is  
\begin{gather}
C = \otimes_{(i,j) \in L} \text{CZ}_{(i,j)}
\end{gather}
and define the corresponding channel
\begin{gather}
\mathcal{C}(\rho) = C \rho  C
\end{gather}
$C$ is self-adjoint, so this is conjugation by $C$. Then a circuit with $d$ repetitions of our noisy layer corresponds to 
\begin{gather}
    \label{eq:cycling_expr}
    \left(\mathcal{N} \circ \mathcal{C}\right)^d
\end{gather}
where the exponent indicates functional iteration. To simplify this expression, note that $\mathcal{C}$ is self-inverse, so we may write 
\begin{align}
    \mathcal{C} \circ \mathcal{N} &= \left(\mathcal{C} \circ \mathcal{N} \circ \mathcal{C}\right) \circ \mathcal{C} \\
    &= \mathcal{N}_C \circ \mathcal{C}
\end{align}
where we define $\mathcal{N}_C \equiv \mathcal{C} \circ \mathcal{N} \circ \mathcal{C}$. From this definition we can also see
\begin{gather}
    \label{eq:commute_rule_1}
    \mathcal{C} \circ \mathcal{N}_C = \mathcal{N} \circ \mathcal{C}
\end{gather}
Consider the action of $\mathcal{C}$ on a single Lindblad term: 
\begin{align}
\left[\mathcal{C} \circ \mathcal{N}_{P,w}\right](\rho) &= C\left[(1 - w)(\rho) + wP (\rho) P \right]C \\
&= \left[(1 - w) C \rho C + w (C P C) C \rho C (C P C)^\dagger\right]
\\ 
&=  \left[\mathcal{N}_{\mathcal{C}(P),w} \circ \mathcal{C}\right] (\rho)
\end{align}
Because $C$ is a Clifford circuit, $\mathcal{C}(P)$ is some Pauli string, and so the resulting channel is a single Pauli-Lindblad term with the same rate and a different generator. The new generator may be less local than the old. In particular, when $C$ is a single layer of two-site gates, the support of $CPC^\dagger$ is up to twice as large as that of $P$. Generalizing this argument to a channel with many Lindblad terms, we see  $\mathcal{N}_C$ is a new Pauli-Lindblad channel with each generator conjugated by $C$. It follows that 
\begin{gather}
    \label{eq:commute_rule_2}
    \mathcal{N} \circ \mathcal{N}_C = \mathcal{N}_C \circ \mathcal{N}
\end{gather}
since all Pauli channels commute.

Applying Equations~\eqref{eq:commute_rule_1} and~\eqref{eq:commute_rule_2} to Expression~\eqref{eq:cycling_expr} leads eventually to the simplification
\begin{gather}
\label{eq:noise_commuted}
\left(\mathcal{N} \circ \mathcal{C}\right)^d = \begin{cases} \mathcal{N}^{\frac{d+1}{2}} {\mathcal{N}_C}^{\frac{d-1}{2}} \circ \mathcal{C} & d\text{ odd} 
\\
\mathcal{N}^{\frac{d}{2}} {\mathcal{N}_C}^{\frac{d}{2}} & d\text{ even} 
\end{cases}
\end{gather}

Equation~\eqref{eq:noise_commuted} is now a single-layer circuit (possibly trivial) followed by a Pauli-Lindblad channel, so we can map it to a Bayesian network exactly as before. This reduction is illustrated in Figure \ref{fig:cycling_reduction}. 

The relationship between the Lindblad rates and the parameters of these new channels is somewhat complicated, but straightforward to compute. The resulting Bayesian network has up to four-site correlations. This can be collapsed into a 1D network and evaluated efficiently\cite{chang_node_1989}. The maximum-likelihood framework also allows us to synthesize data from multiple cycle depths together. This is especially useful when it is necessary to learn SPAM errors. 

\begin{figure}[h]
    \centering
    \includegraphics[width=\linewidth]{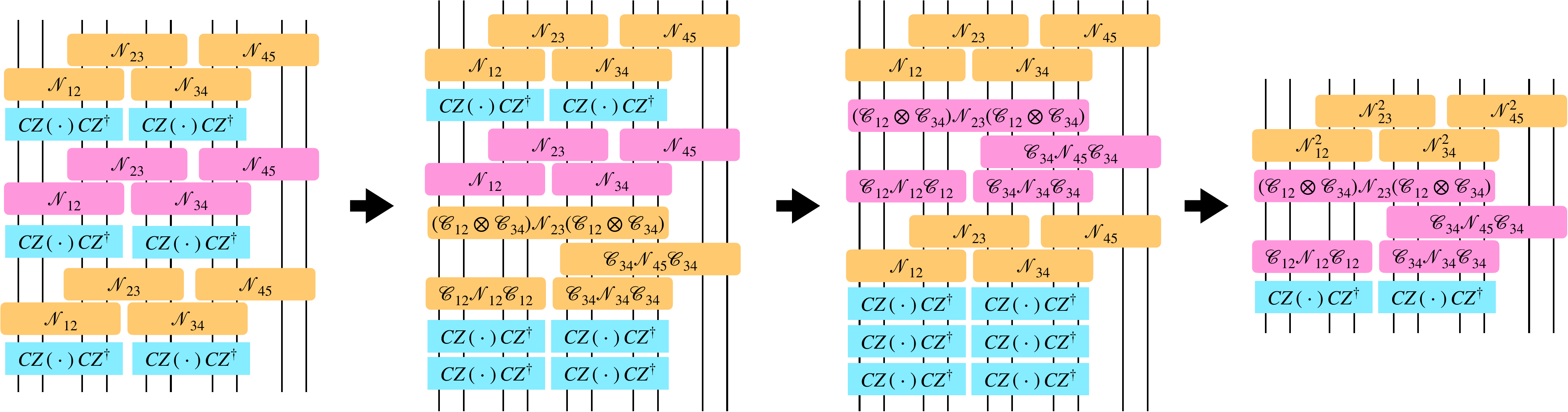}
    \caption{Simplifying the multi-cycle case. As we commute the $CZ$ gates through the errors, we use the fact that $CZ$ is self-inverse to cancel out repeated copies. The resulting channel has at most 4-site Lindblad generators no matter how many cycles are involved.}
    \label{fig:cycling_reduction}
\end{figure}

\subsection{Measurement error}
\label{sec:measurement_error}
In real circuits, state preparation and measurement (SPAM) errors are quite significant\cite{van_den_berg_probabilistic_2023}. We will now show that there is a natural and straightforward way to include these errors in the computational model. 

Here we will focus on the case of measurement errors. State preparation errors can be incorporated using a very similar approach. We may view measurement error as the probability of a Pauli X error immediately before the (Z-basis) measurement. We assume this probability may be different for each qubit, but does not depend on the Pauli basis rotation or number of layers of Clifford gates applied. These errors can be represented by one-site Pauli-Lindblad channels, as illustrated in Figure \ref{fig:spam}a. One may thus compute the likelihood function by the same methods as discussed above. 

In order to learn SPAM error rates and gate-based error rates simultaneously, one needs access to measurements at multiple circuit depths. Data from only a single depth is not sufficient to disambiguate SPAM and internal error parameters. Previous approaches to identifying SPAM have attempted to characterize it separately from gate-based noise\cite{van_den_berg_model-free_2022, magesan_characterizing_2012, helsen_new_2019}. A maximum likelihood framework, on the other hand, can naturally learn SPAM errors and gate-based errors concurrently. The only requirement is that the data be informationally sufficient to disambiguate internal and SPAM error rates\cite{chen_disambiguating_2026}. 

Figure \ref{fig:spam}b shows results on $10$ qubits. The circuit in this case is simply idling. The noise channels are again generated randomly. Internal error probabilities are chosen from the uniform distribution between $0$ and $10^{-3}$, while measurement error probabilities are chosen from the uniform distribution between $0$ and $10^{-2}$. We take $5$ million samples each at depths $0$ and $50$. With this data, we see that the maximum likelihood estimator learns parameters which are generally close to the true values. Relative accuracy is worse for very small parameters, perhaps because the corresponding errors occur only very rarely in the dataset. 

\begin{figure}[h]
    \centering
    \begin{tikzpicture}
        \begin{scope}
            \node[anchor=north west,inner sep=0] (image_a) at (0,0)
            {\includegraphics[width=0.4\columnwidth]{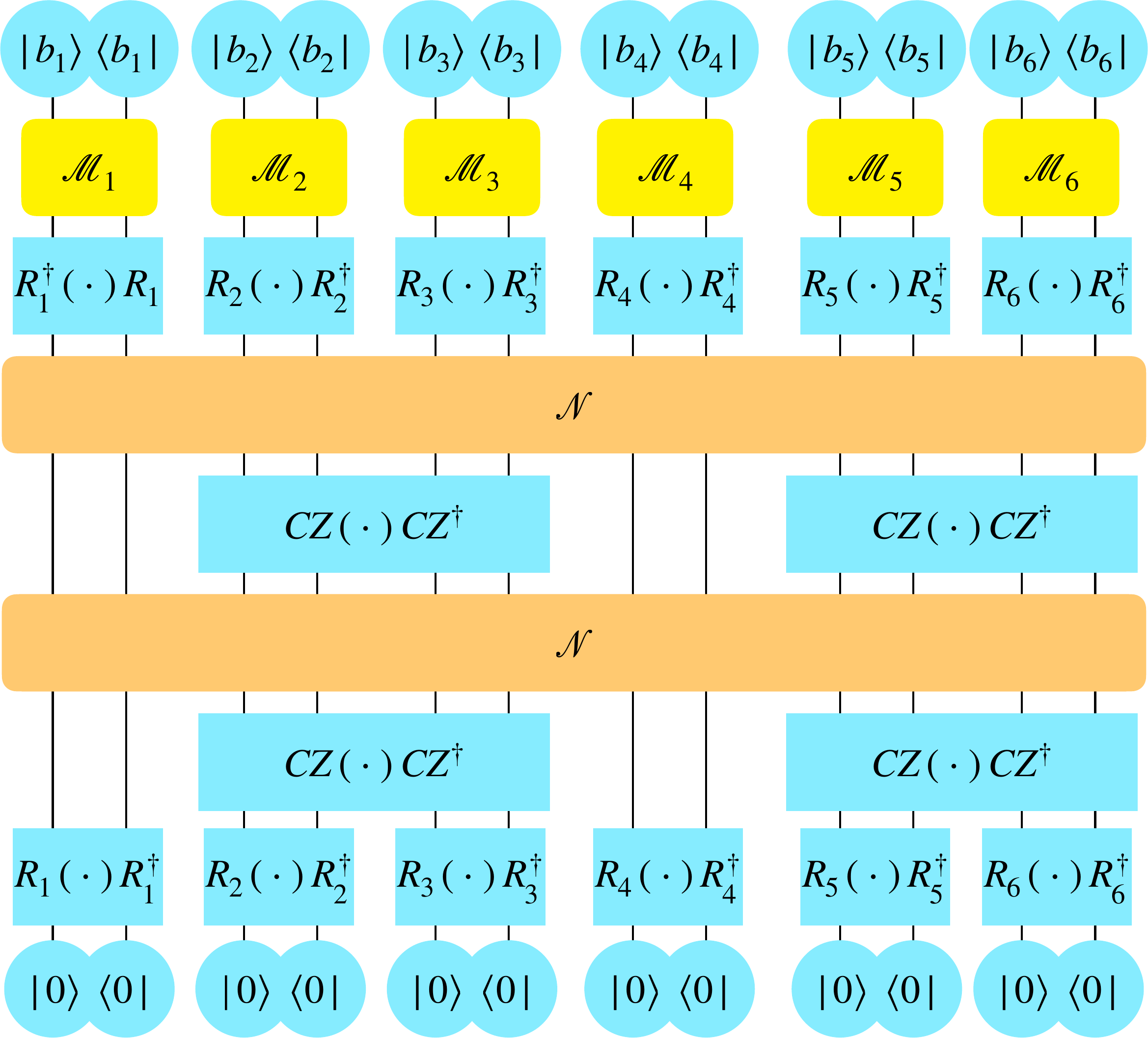}};
            \node [anchor=north west] (note) at (-0.5,0) {\small{\textbf{a)}}};
        \end{scope}
    \end{tikzpicture}
    \begin{tikzpicture}
        \begin{scope}
            \node[anchor=north west,inner sep=0] (image_b) at (0,0)
            {\includegraphics[width=0.5\columnwidth]{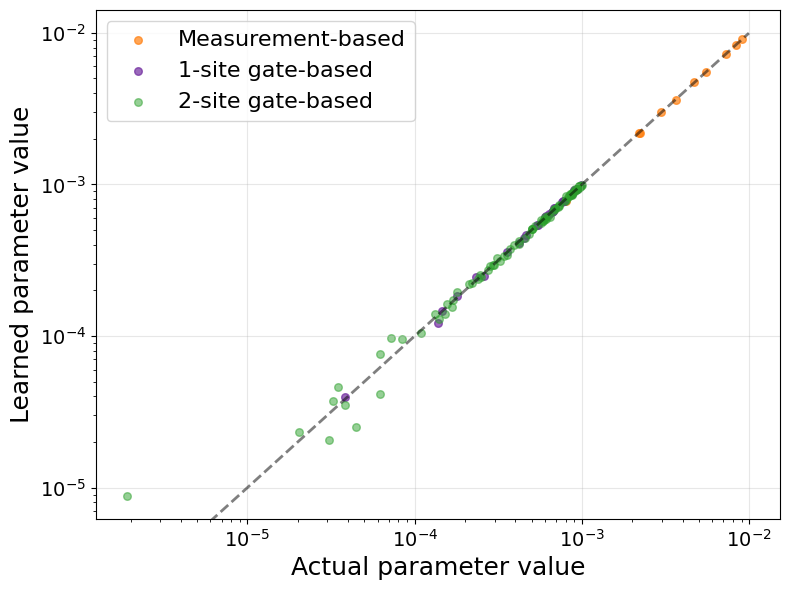}};
            \node [anchor=north west] (note) at (-0.5,0.3) {\small{\textbf{b)}}};
        \end{scope}
    \end{tikzpicture}
    \vspace*{-0.4cm}
    \caption{(a) Circuit model used for learning SPAM. Measurement errors $\mathcal{M}$ correspond to bit flips immediately before measurement. Note that the effective channel $\mathcal{N}$ depends on the cycle count, as shown in Figure \ref{fig:cycling_reduction}. (b) Successful learning of measurement error on 10 qubits. We see that learned estimates of both gate-based error and SPAM error parameters are close to the true values, with larger relative errors in small parameters. $5 \times 10^{6}$ training samples were taken at cycle counts 0 and 50. The channel is sampled randomly with maximum measurement error rate $10^{-2}$ and maximum gate-based error rate $10^{-3}$.
    }
    \label{fig:spam}
\end{figure}

\section{Discussion}
We have given an efficient noise-learning algorithm using MLE. MLE is known to be asymptotically optimal; we have shown that its advantage over other available methods is large enough to be important in practice. There are a number of possible extensions of our approach which may be valuable directions for future work. 

\subsection{Generalizing beyond 1D}
Our method for efficient exact evaluation of the likelihood function relies on a 1D (or at least treelike) connectivity pattern. The typical chip layout, however, is 2D. We propose three possible strategies for non-1D geometries. None are implemented in this work, but they offer promising approaches for future study. 

The first and simplest strategy is to cut the network into treelike subnetworks. This can be done by freezing certain latent variables, then summing explicitly over the possible configurations of those frozen latent variables. This method gives an exact value. It requires exponential time, but error mitigation itself is also only practical for modest system sizes, and so the classical overhead may not be limiting in practice. For some chip layouts (e.g., heavy hex) one needs to ``cut'' only a small fraction of the edges in order to obtain a decomposition into treelike subgraphs.

With 2D geometric locality, a second possibility is what we term ``patching'', inspired by Ref~\onlinecite{napp_efficient_2022}. This exploits the fact that correlations are quite short-range. The basic idea is to choose a distance threshold $\zeta$ beyond which patches are assumed to be independent. One may then divide a 2D lattice into a checkerboard pattern. We first compute unconditional probabilities for a set of patches separated by a distance $\zeta$. The next step is to fill in the patches between them, computing probabilities conditional only on the neighboring patches which have already been filled in. We expect the precision of this algorithm to be good so long as the conditional mutual information between nonadjacent patches is small. It seems plausible that this condition holds for most noise models, but further study is needed.

A third option is the standard belief propagation algorithm\cite{yedidia_understanding_2003}. Belief propagation is exact for treelike connectivity, but approximate for graphs with loops. Suppose for a moment that our errors remain $2$-local. For the heavy hex lattice, the shortest possible loop then involves 6 independent latent variables (in the case where 6 CZs are arranged around the same hex). In the actual data, the probability that these correlation-inducing errors will occur concurrently is approximately $(10^{-3})^{6}$, so to learn any single local parameter we may effectively truncate our interactions to a locally treelike neighborhood. At larger circuit depths the situation is a bit worse. Appendix \ref{app:fisher}, however, shows that even at the optimal depth, each individual local error only occurs with probability $\sim 0.017$. This raises the hope that belief propagation, too, may give a very good approximation to the likelihood.

\subsection{Generalization to other gates and channels}
Although the Bayesian network representation is conceptually and computationally useful, it depends heavily on the particular structure of this problem. Tensor networks offer a more flexible approach to computing the likelihood function\cite{torlai_quantum_2023, mangini_tensor_2024, cao_differentiable_2026}. In the 1D case they can handle arbitrary non-Pauli channels, so long as the channel can be decomposed into a product of local channels as in Figure \ref{fig:network_reductions}a. Furthermore, rather than learning the errors associated with a single-layer CZ circuit, a tensor network can incorporate any efficiently contractible circuit (e.g., multilayer shallow 1D circuits with non-Clifford gates). 

On the other hand, generalizing the tensor network algorithm to non-1D systems seems more challenging. Furthermore, the constant-factor computational overhead of constructing and contracting the tensor network for each training sample is relatively high. 

\subsection{Model pruning and goodness-of-fit testing}
We have assumed that only Lindblad generators below some geometric locality cutoff are important elements of the noise model. In a real system, this may not always be the case\cite{govia_bounding_2025}. It is interesting to ask whether MLE allows us to detect the presence of unmodeled error terms. Some common tools include the Akaike or Bayesian Information Criterion or the Generalized Likelihood Ratio Test, both of which are made possible with access to the likelihood function. 

If any specific additional error term is suspected, these methods provide a strategy for deciding if it should be included. But because the number of possible generators is exponentially large, there is an additional challenge: One must somehow decide which extensions of the model are likely to be worth considering. One possible strategy is to look for training samples which are especially surprising under the current model. One can then attempt to construct Lindblad generators that could explain those particular training samples. This may be a useful direction for future work.

\subsection{Improved inversion}
For PEC, overestimating an error rate is more harmful than underestimating it. The impact of an overestimate grows exponentially with circuit depth, while an underestimate merely fails to correct some of the error. It may thus be beneficial to choose a model for PEC with error rates which are a bit smaller than the point estimates learned from the data. The Fisher metric, which can be computed from the likelihood function, could lead to a systematic approach to this problem.

\begin{acknowledgments}
We acknowledge funding from the IBM-Illinois Discovery Accelerator Institute and useful discussions with Kristan Temme.
\end{acknowledgments}

\bibliography{references}

\appendix
\section{Empirical Pauli Fidelities are not a sufficient statistic}
\label{app:insufficient}
We will show that the empirical Pauli fidelities do not capture all of the important information in a dataset. In other words, they are a lossy summary of the observed data. The formal version of this statement is that they are not a \textbf{sufficient statistic}, which is to say that the likelihood is not functionally dependent on the EPFs. 

\begin{theorem}
    Consider a channel with Lindblad generators $\mathcal{P}$ and parameters $\mathbf{w}$. The empirical Pauli fidelities $\hat{f}_P : P \in \mathcal{P}$ are not necessarily a sufficient statistic for $\mathbf{w}$. 
\end{theorem}
\begin{proof}
    Consider \(\mathcal{P} = \{XX, ZZ, XZ\}\). We will construct a pair of datasets $\mathcal{D}_1, \mathcal{D}_2$, which give the exact same Pauli fidelities, but different likelihood functions. 
    
    We take one measurement in each basis. The measurement outcomes are given in Tables \ref{tab:d1} and \ref{tab:d2}. From these outcomes we compute empirical fidelities. The general formula for empirical Pauli fidelities is 
\begin{gather}
    \hat{f}_{P}(\mathcal{D}) = \frac{1}{|M_{P}|}\sum_{\mathbf b \in M_{P}} (-1)^{\sum_{i \in \text{supp}(P)} b_i}
\end{gather}
where  $M_P$ is the set of measurements in bases containing $P$ as a substring. In this example all Pauli strings have full support and there is only one measurement per basis, so the formula is somewhat simpler. Applying it, we discover that the two datasets give exactly the same fidelity estimates: $f_P(\mathcal{D}_1) = f_P(\mathcal{D}_2)$ for all $P \in \mathcal{P}$. 

\begin{table}[h]
\begin{subtable}
    \centering
    \begin{tabular}{cc c}
        \toprule
        \textbf{String} & \textbf{Syndrome} & \textbf{Fidelity} \\
        \midrule
        $XX$ & $00$ & $1$ \\
        $XZ$ & $00$ & $1$ \\
        $ZZ$ & $10$ & $-1$ \\
        \bottomrule
    \end{tabular}
    \caption{Data set $\mathcal{D}_1$}
    \label{tab:d1}
\end{subtable}
\vspace{2em}
\begin{subtable}
    \centering
    \begin{tabular}{cc c}
        \toprule
        \textbf{String} & \textbf{Syndrome} & \textbf{Fidelity} \\
        \midrule
        $XX$ & $00$ & $1$ \\
        $XZ$ & $00$ & $1$ \\
        $ZZ$ & $01$ & $-1$ \\
        \bottomrule
    \end{tabular}
    \caption{Data set $\mathcal{D}_2$}
    \label{tab:d2}
\end{subtable}
\end{table}

It should be clear that the two datasets do not indicate the same channel behavior. Under $\mathcal{D}_1$, the syndrome seen on our $ZZ$ measurement implies that an $XZ$ error occurred. Under $\mathcal{D}_2$, on the other hand, the syndrome seen on our $ZZ$ measurement implies that $XZ$ and $XX$ errors must have occurred together. So we should assign a larger Lindblad rate to $XX$ if we see $\mathcal{D}_2$ rather than $\mathcal{D}_1$. The $\hat{f}$ fail to capture this difference.

To formalize this intuition, consider parameters
\begin{gather}
    \begin{bmatrix}
        w_{XX} \\
        w_{XZ} \\
        w_{ZZ} \\
    \end{bmatrix} = \begin{bmatrix}
        0 \\
        \frac{1}{3} \\
        0 \\
    \end{bmatrix} \equiv \mathbf{w}^*
\end{gather}
Under $\mathcal{D}_1$, we know an $XZ$ error did not happen in the first two trials and did happen in the third trial, so
\begin{gather}
    \mathcal{L}(\mathbf{w}^* | \mathcal{D}_1) = \left(\frac{2}{3}\right)^2 \cdot \frac{1}{3}
\end{gather}
But under $\mathcal{D}_2$, we can be certain that trial 3 involved both an $XX$ and an $XZ$ error, which never occurs under these parameters. So
\begin{gather}
    \mathcal{L}(\mathbf{w}^* | \mathcal{D}_2) = 0
\end{gather}
This is an example of two cases with the same $\hat{f}$ but distinct $\mathcal{L}$, so the empirical Pauli fidelities are not a sufficient statistic. 
\end{proof}

This example is rather contrived. However, one can construct similar examples for a channel with all possible 1D $2$-local generators included. It seems to be generally true that empirical Pauli fidelities do not capture all of the important information present in the data.

\section{Optimal shot allocation}
\label{app:fisher}
Section \ref{sec:cycling} shows that it is possible and useful to use data from deeper circuits. One may now ask which circuit depths are most informative. More precisely, given a finite sample budget, how should samples be allocated across various circuit depths? The maximum likelihood approach offers a natural framework in which to study this question. 

Previous experiments have typically allocated shots uniformly over some range of depths. We will show that this is not the best choice. PEC requires very precise tomography, and so we are interested mostly in the limit of large sample size. In this limit the MLE is asymptotically normal, with a covariance matrix given by the inverse of the Fisher metric\cite{kwon_universal_2026}. The Fisher metric may be computed from the likelihood function by 
\begin{gather*}
    F_{ij} = \mathbb{E}_{\mathbf{x}}\left[- \frac{\partial^2 \log \mathcal{L}(w_1...w_m|\mathbf{x})}{\partial w_i \partial w_j}\right]
\end{gather*}
Suppose we have $n_{d}$ shots taken in each measurement basis $b$ at circuit depth $d$. Then the Fisher information may be written more explicitly as
\begin{gather*}
    F_{ij} = \sum_{d \in \text{depths}} \sum_{b \in \text{bases}} n_{d} \mathbb{E}_{\mathbf{x} \sim d,b}\left[- \frac{\partial^2 \log \mathcal{L}^{(d,b)}(w_1...w_m|\mathbf{x})}{\partial w_i \partial w_j}\right]
\end{gather*}
The expected mean squared error of the learned parameters is given by $\frac{\tr F^{-1}}{\text{\# params}}$. We then search for the distribution of shots which minimizes this error, subject to a constraint on the total shot budget. Figure \ref{fig:var_vs_depth}a shows how the expected mean squared error depends on the depth at which data is taken, for data taken only at a single depth. 

\begin{figure}[h]
    \centering
    \begin{tikzpicture}
        \begin{scope}
            \node[anchor=north west,inner sep=0] (image_a) at (0,0)
            {\includegraphics[width=0.6\columnwidth]{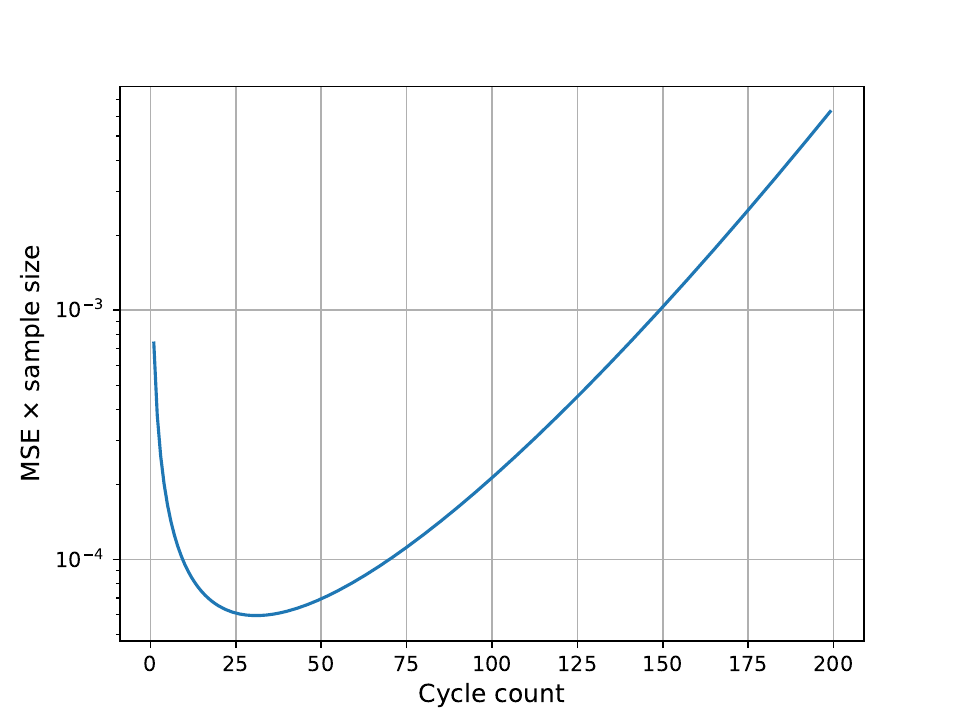}};
            \node [anchor=north west] (note) at (-0.5,0) {\small{\textbf{a)}}};
        \end{scope}
    \end{tikzpicture}
    \vspace*{-0.4cm}
    \caption{Typical error of our estimated parameters vs. circuit depth at which data is collected. We use the Fisher information to compute the asymptotic MSE. The amount of information about the channel parameters extractable from the data improves rapidly at small depths, but deteriorates at large depth as the effective error rates all approach $\frac{1}{2}$. This data is for a 10-qubit channel with $w_i$ drawn from the uniform distribution over $(0, 10^{-3})$.
    }
    \label{fig:var_vs_depth}
\end{figure}

We find empirically that the optimal choice is always to allocate the entire budget to a single depth. In the absence of SPAM errors, there is no apparent advantage to spreading data across multiple depths. Furthermore, this depth seems to scale roughly as the inverse of the average basis error probability, roughly $\left(41 \cdot \frac{1}{\text{\# params}}\sum_i w_i\right)^{-1}$. This is shown in Figure \ref{fig:optimal_depth_details}a. Figure \ref{fig:optimal_depth_details}b compares this approach against the more standard strategy of distributing shots uniformly over some range of depths. We see that the single-depth strategy is slightly better, but more sensitive to the chosen depth.

\begin{figure}[h]
    \centering
    \begin{tikzpicture}
        \begin{scope}
            \node[anchor=north west,inner sep=0] (image_a) at (0,0)
            {\includegraphics[width=0.43\columnwidth]{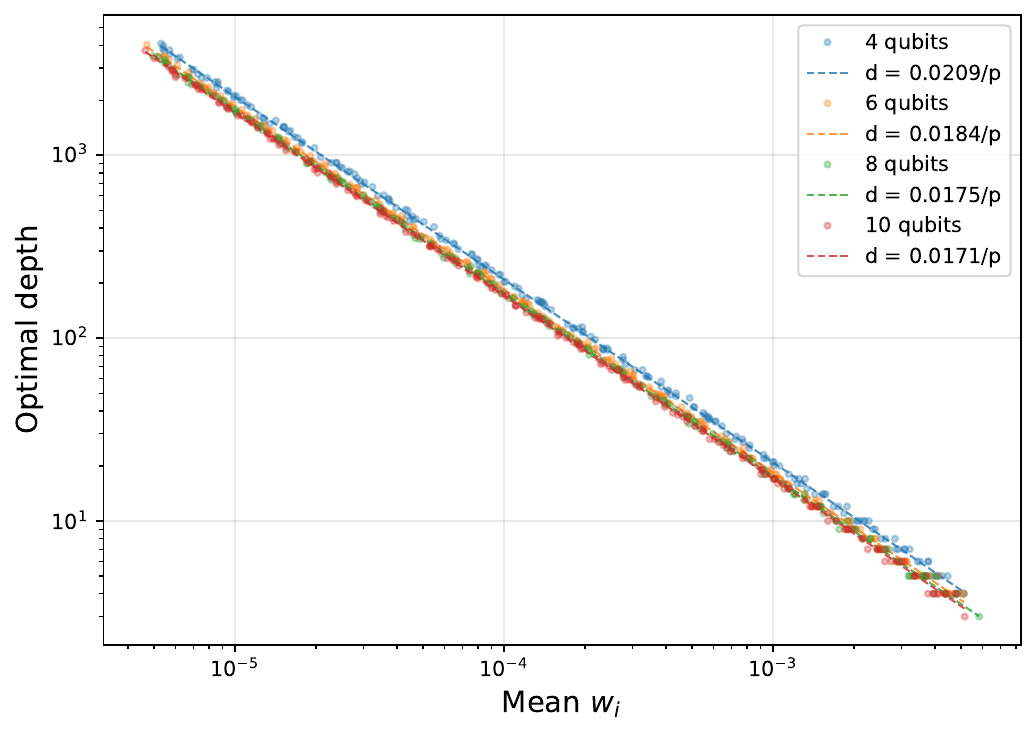}};
            \node [anchor=north west] (note) at (-0.5,0) {\small{\textbf{a)}}};
        \end{scope}
    \end{tikzpicture}
    \begin{tikzpicture}
        \begin{scope}
            \node[anchor=north west,inner sep=0] (image_b) at (0,.5)
            {\includegraphics[width=0.46\columnwidth]{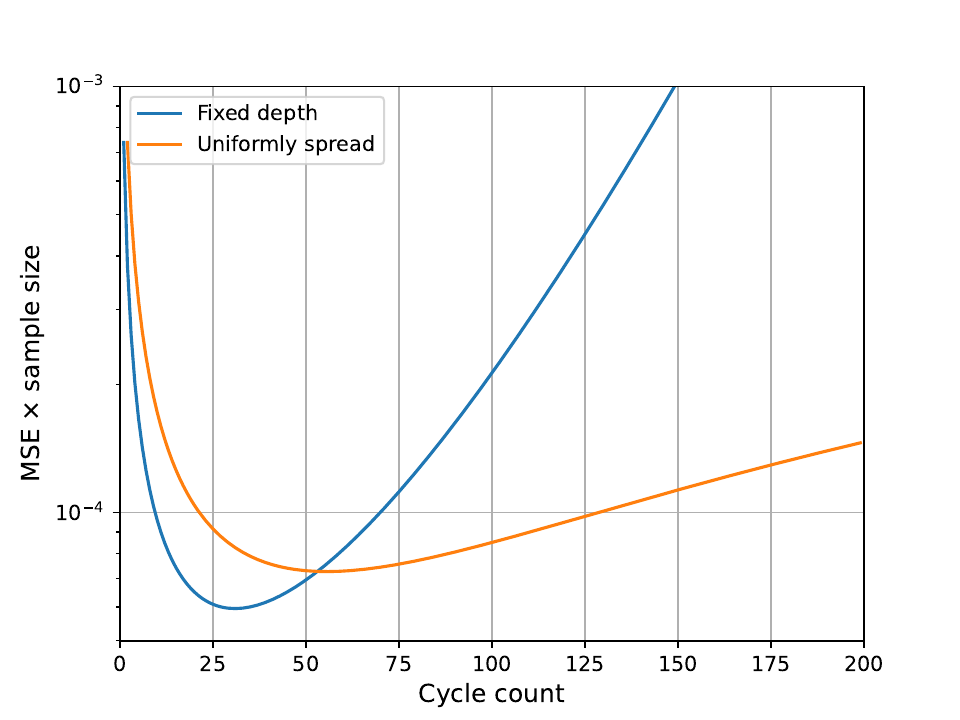}};
            \node [anchor=north west] (note) at (-0.5,0) {\small{\textbf{b)}}};
        \end{scope}
    \end{tikzpicture}
    \vspace*{-0.4cm}
    \caption{(a) Most informative depth at which to take data for many channels with $w_i$ sampled from uniform distributions. We observe an optimal depth generally close to $\frac{0.017}{\text{typical error rate}}$, with slightly deeper depths better for small systems. (b) Collecting data at a single depth $d$ vs. at depths spread uniformly between $0$ and $d$. The single-depth strategy is slightly better, but the uniform strategy is less sensitive to the choice of cutoff.
    }
    \label{fig:optimal_depth_details}
\end{figure}

There are two difficulties with this approach. First, the true parameters are unknown, as is $F$. However, we can first use some shots allocated heuristically to obtain a rough estimate of the true parameters. We can then use this rough estimate to decide how to allocate our remaining shots. In practice our tests are run with $10\%$ of the shot budget used in the former phase and $90\%$ in the latter. The second difficulty is that the expected value is expensive to compute exactly. It is necessary to instead estimate the elements of $F$ by a Monte Carlo procedure. This introduces some noise and classical overhead, but is expected to converge reasonably quickly.

\subsection{Shot allocation with measurement error}
Simultaneous learning of measurement error rates and internal error rates requires data from multiple depths. We may again use the Fisher information to determine how a fixed shot budget ought to be allocated across depths. 

The natural choice of cost metric here is not so simple. For a single layer of gates, measurement errors are typically much more common than internal error parameters, which suggests that it is more important to correct measurement error. But for a deep circuit with PEC, errors in our estimates of the internal parameters will accumulate, leading to large errors in the output. This suggests that one ought to care more about tomography of measurement parameters when the eventual circuit of interest is shallow, but care more about internal parameters when the target circuit is sufficiently deep. 

Here we do not resolve this tradeoff. We instead consider cost functions that are a weighted sum of the expected mean squared error of SPAM parameter estimates and the expected mean squared error of internal parameter estimates. We may then study the Pareto frontier of shot allocation choices in this 2D plane. This is shown in Figure \ref{fig:pareto}. In the limit of asymptotic normality, the cost function can be computed from the Fisher matrix. This makes it numerically tractable to find the optimal distribution of data points over various depths. The data shown is only for idling circuits, with no gates applied, but the generalization to a repeated layer of Clifford gates is again straightforward.

\begin{figure}[h]
    \centering
    \includegraphics[width=0.6\linewidth]{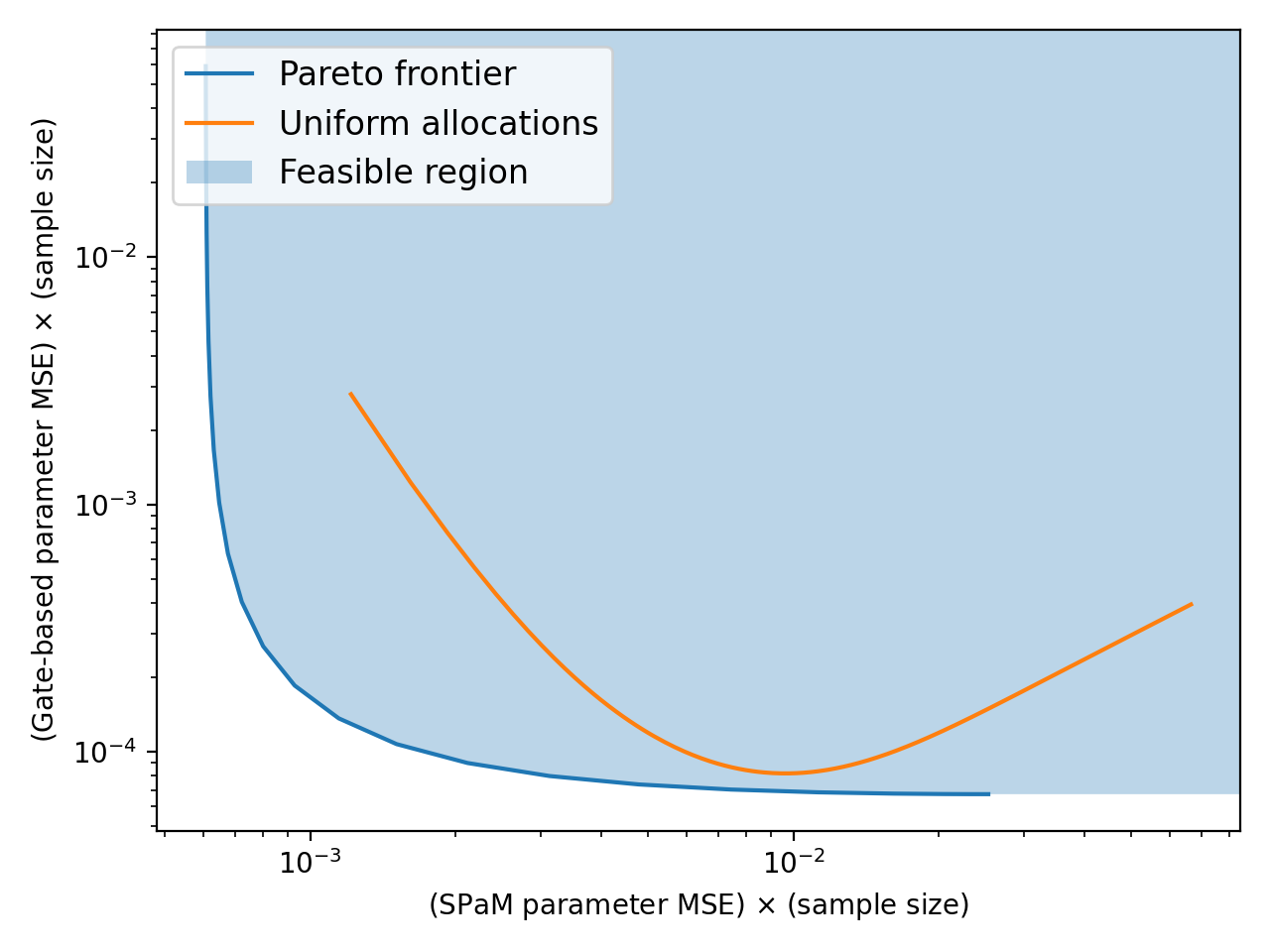}
    \caption{Feasible set for shot allocation. Allocating shots uniformly over some range of depths, as has been done in previous works, results in errors $\sim 10 \times$ larger than choosing an allocation on the Pareto frontier. This data is for a random 10-qubit channel with maximum error probability $\sim 10^{-3}$.}
    \label{fig:pareto}
\end{figure}

We find empirically that the optimal allocation involves data taken at only two depths, one of which is zero. In other words, it is best to first use some of the shot budget to learn the measurement error, with internal errors ``turned off'' (since no gates are applied). The rest of the shot budget can then be used at a particular large depth in order to learn the internal parameters. The more traditional strategy of allocating shots uniformly over some range of depths does not intersect the Pareto frontier.

\end{document}